\documentclass[letterpaper, 10 pt, conference]{ieeeconf}  

\IEEEoverridecommandlockouts                              

\usepackage[utf8]{inputenc} 
\usepackage[T1]{fontenc}
\usepackage{url}
\usepackage{ifthen}
\usepackage{cite}
\usepackage{amsmath,amssymb} 
\usepackage{color,graphicx,cite, soul}
\newtheorem{theorem}{Theorem}

\newtheorem{claim}{Claim}

\newtheorem{conjecture}{Conjecture}
\title{\LARGE \bf
On Distributed Storage  Allocations of Large Files\\ for Maximum  Service  Rate
}

\author{Pei Peng$^{1}$ and Emina Soljanin$^{1}$
\thanks{$^{1}$Pei Peng and Emina Soljanin are with the Department of Electrical and Computer Engineering, Rutgers University, Piscataway, NJ 08854 USA.
        {\tt\small pei.peng@rutgers.edu; emina.soljanin@rutgers.edu}}%
}

\begin{document}
\maketitle

\begin{abstract}
Allocation of (redundant) file chunks throughout a distributed storage system affects important performance metrics such as the probability of file recovery, data download time, or the service rate of the system under a given data access model. 
This paper is concerned with the service rate under the assumption that the stored data is large and its download time is not negligible.
We focus on quasi-uniform storage allocations and provide a service rate analysis for two common data access models. We find that the optimal allocation varies  in  accordance  with  different  system  parameters. This was not the case under the assumption that  the  download  time  does  not scale with the size of data, where the
minimal spreading allocation was previously found to be universally  optimal. 

\end{abstract}

\section{INTRODUCTION}
Distributed storage systems (DSSs) are, in various guises, an integral part of different computing and content providing environments such as cloud data centers, caching edge networks, and more recently, fog systems. Their purpose is to ensure reliable storage and/or quick access of data by end users or computing processes. Today, both goals are being increasingly addressed by storing data redundantly, either by replication or erasure coding. This paper is concerned with allocations of redundant data chunks throughout a DSS that ensure maximum data access service rate.

Most of the work on data access in DSSs is concerned with the download latency (see e.g., \cite{joshi2012coding, chen2014queueing,tandon2014new,kadhe15availability} and references therein). It has recently been recognized, that another important metric that measures the availability of the stored data is the service rate \cite{noori2016storage,ServiceRate:AktasAJJKMMS}. Maximizing service rate (or the throughput) of a distributed system helps support a large number of simultaneous system users. Rate-optimal strategies are also latency-optimal in high traffic. Thus, maximizing the service rate also reduces the latency experienced by users  in particular in highly contending scenarios. 

This paper adopts a DSS model originally proposed in \cite{leong2011distributed}. In this model, a file is split into multiple chunks, and (replication or coded) redundancy is introduced at some fixed level determined by the storage budget that the DSS has for the file. This total storage is the only constraint, and there is no limit on how many chunks a particular node can store as long as it stays within the budget. Attempts to data retrieval are done according to some limited access models.

Several studies have looked into how to allocate redundant chunks of data over the storage nodes, focusing mostly on optimizing two DSS performance metrics \cite{leong2011distributed,leong2012distributed,Sardari_Allocation_2010,hong2014asymptotic}. One of them is the {\it probability of successful data recovery} when only a subset of (possibly failed) nodes are  accessed, and the other is the {\it average download time}  when a set of nodes from which the file can certainly be recovered is accessed.  Finding these quantities has shown to be quite challenging, and optimal allocations are known only in some special cases. Some versions of this problem are related to a long standing conjecture by Erd\H{o}s on the maximum number of edges in a uniform hypergraph \cite{matching:AlonFHRRS12}.

In general, both measures are of interest and should be simultaneously taken into account. Often increasing the chance of successfully downloading a file while desirable should not come at the cost of intolerable delivery delay. Moreover, in practice, we may often want to partially sacrifice a successful but tardy data delivery to some users in order to ensure that other users, that can receive the data, are indeed served fast. 

Note that, depending on the allocation, some subsets of nodes may not contain enough file chunks between them to ensure data recovery, and accessing them will result in a zero system's service  rate. On the other hand, again depending on the allocation, some subsets of nodes will contain redundant file chunks, and that redundancy (superfluous for file recovery) can be exploited to increase the service rate. 
These issues were first addressed in \cite{noori2016storage}, where a non intuitive conclusion was reached that the allocation that maximizes the probability of successful data recovery is often not the one that maximizes the average service rate.  
%

Depending on the number of storage nodes and the allocated redundancy budget, it may be beneficial for recovery to maximally spread the redundant file chunks
over the nodes, whereas concentrating the redundant chunks (minimum spreading) may increase the expected service rate. 
The work of \cite{noori2016storage} assumes that the download time 
is random because of independent workload fluctuations inherent to the system, and does not depend on the size of the data being downloaded. We here assume that the stored data is large, its download time is not negligible and scales with the size of the data. We find that the optimal allocation varies  in  accordance  with  different  system  parameters, which was not the case in \cite{noori2016storage}, where the minimal spreading allocation was found to  be  universally  optimal.


The paper is organized as follows. A DSS model and problem formulation are given in Sec.~\ref{Sec:model}. Service rate analysis considering the effect of access model and the success of serving a request is presented in Sec.~\ref{Sec:analysis}.  Some numerical examples and further discussion are provided in Sec.~\ref{Sec:simulation}.

\section{SYSTEM MODEL}
\label{Sec:model}
\begin{figure}[t]
    \centering
    \includegraphics[width=0.4\textwidth]{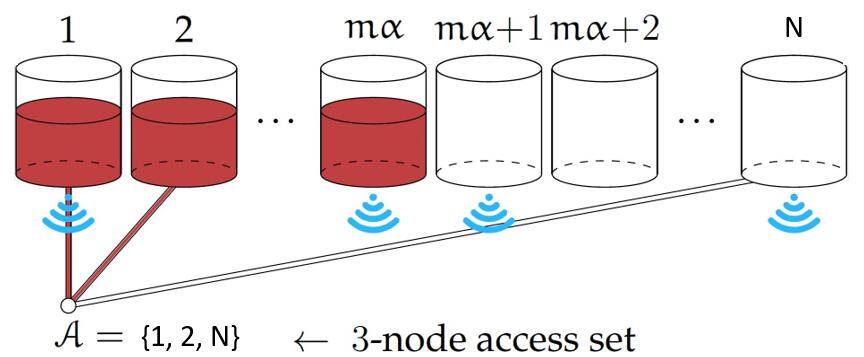}
    \caption{
    A DSS of $N$ nodes where each node stores either $k/\alpha$ or $0$ data blocks of interest to some users, and thus only $\varphi=\alpha m$ nodes contain data blocks. The WiFi sign indicates that the node has enough  available capacity to serve the user. Note that that is independent of whether or not the node has been accessed or has the data. Here, three nodes are successfully accessed, but only two of them have (coded) data blocks. One of the accessed node has data blocks but is not able to serve the user.}
    \label{fig:DSS}
\end{figure}

\subsection{Storage Model\label{sec:StorageModel}}
A file consisting of $k$ blocks is to be redundantly stored over a  DSS with $N$ storage nodes.  To protect the data against nodes' failure or unavailability, the file is encoded by an MDS code into $mk$ ($m \in \mathbb{N}$) encoded blocks so that any $k$ of them are sufficient to recover the original file. The $mk$ encoded blocks are partitioned into $N$ subsets $\mathcal{S}_i$'s for $i  \in \{1,\dots, N\}$ where $| {\mathcal S}_i| = s_i$, and thus $\sum_{i = 1}^N s_i = mk$.   We refer to such partitioning an {\it allocation}. The $s_i$ blocks in $\mathcal{S}_i$ are stored at the storage node $i$. 
 
We are concerned with quasi-symmetric  allocations \cite{Sardari_Allocation_2010}, where a node can either store a constant number of  blocks $k/\alpha$($\alpha \in \mathbb{N}$) or no blocks at all.
(Dealing with a general storage allocation optimization problem is computationally difficult for a general setup, see e.g., \cite{leong2012distributed}.)
We will refer to such allocations as $\alpha$ quasi-symmetric allocation. Fig.~\ref{fig:DSS} depicts an example quasi-symmetric allocation on $N$ nodes.

We refer to a quasi-symmetric allocation where $\alpha = 1$ as \emph{minimal spreading} \cite{leong2012distributed}. Note that for a minimal spreading allocation, the $k$ file blocks are simply replicated over some $m$ storage nodes. Similarly, an allocation with $\alpha = N/m$ will be referred to as a \emph{maximal spreading} allocation since the file chunks are spread over all $N$ nodes in the system.

\subsection{Data Access and Delivery Models\label{sec:AccessModel}}
\noindent\ul{Fixed-size Access:}
In this model, the download request is forwarded to a random $r$-node subset of the $N$ storage nodes \cite{leong2012distributed, Sardari_Allocation_2010}. Therefore, the access to a given $r$-subset $\mathcal{A}$ results in the successful recovery of the data iff the nodes in $\mathcal{A}$ jointly contain at least $k$ coded blocks: 
\begin{equation}\label{eq: Recovery cond}
\sum_{i \in \mathcal{A}} s_i \ge k.
\end{equation}
Note that for $\alpha > r$, it is impossible to recover the data. Thus, we only consider the $1 \leq \alpha \leq r$ case.
\\[1ex]
\ul{Probabilistic Access:}
In this model, the download request is forwarded to all nodes that store the data. However, the request to a node fails with probability $p$. Assuming that $\mathcal{A}$ represents the set of nodes that are successfully accessed, the condition for data recovery is also (\ref{eq: Recovery cond}). In this case, $1 \leq \alpha \leq \frac{N}{m}$. In this access model, $|\mathcal{A}|$ is a Binomial random number between 1 and $N$.
\\[1ex]
 Regardless of the access model, for an accessed subset of nodes $\mathcal{A}$, we denote the number of nodes containing data by $\varphi(\mathcal{A})$. For instance, in Fig.~\ref{fig:DSS}, three nodes 
($|\mathcal{A}| = \text{3}$) are accessed while only $\varphi(\mathcal{A}) = 2$ of them have data. For an $\alpha$ quasi-symmetric allocation, data recovery from this subset is successful iff $\varphi(\mathcal{A}) \geq \alpha$. 
The probability of successful file recovery under an $\alpha$ allocation is, therefore, given by
\begin{equation}
P_s(\alpha)=
\sum_{\mathcal{A}:\; \sum_{i \in \mathcal{A}} s_i \ge k}
P(\mathcal{A})
\label{eq:prob}
\end{equation}
where $P(\mathcal{A})$ is the probability of acccesing $\mathcal{A}$. Note that the sum goes over all sets $\mathcal{A}$ that satisfy the condition \eqref{eq: Recovery cond}.
  
\subsection{Service Models}
We assume a request is simultaneously served by all nodes in the accessed set $\mathcal{A}$ that contain data, where each node takes some i.i.d.\ random time to deliver its blocks. In the fixed-size access model $|{\mathcal A}| = r$ while in the probabilistic access model, $|\mathcal{A}|$ is a Binomial random variable between 1 and $N$. Note that the file can be reconstructed when the accessed nodes jointly deliver $k$ encoded blocks. We here limit our study to the case where a node has to deliver all its blocks for the download to count. 

For an $\alpha$ quasi-symmetric allocation, the download request can be served iff $\varphi(\mathcal{A}) \geq \alpha$, and as soon as all blocks are downloaded from any $\alpha$ out of the $\varphi(\mathcal{A})$ nodes with data. Therefore, the average download time
  $T_s(\alpha|\varphi(\mathcal{A}))$ is the $\alpha$-th order statistics of $\varphi(\mathcal{A})$ waiting times at the storage nodes.
\\[1ex]
\ul{Scaled Exponential Service:} In this model \cite{joshi2012coding}, a node delivers the first file block in some exponential random time, and  each subsequent block in the the same time. We assume that a node storing the whole file delivers all of its blocks in a random time exponentially distributed with the mean $1/\mu$.
It is easy to see that, equivalently, we can say that a node storing $1/\alpha$ fraction of the file delivers all of its blocks in the random time exponentially distributed with the mean $1/(\alpha\mu)$. (Recall
that in \cite{noori2016storage},
the download times are assumed exponential i.i.d.\
and independent of the size of the data being downloaded.)
For this model, we have
\begin{equation}
  T_s(\alpha|\varphi(\mathcal{A}))=\frac{1}{\alpha\mu}(H_{\varphi(\mathcal{A})}-H_{\varphi(\mathcal{A})-\alpha}),
\label{eq:time}  
\end{equation}
where $H_{\ell}=\sum^{\ell}_{i=1}1/i$ denotes the $\ell$-th 
harmonic number, and $1/(\alpha\mu)$ comes from the service rate scaling discussed above.
The corresponding service rate from set $\mathcal{A}$ (with $\varphi(\mathcal{A})>\alpha$ nodes containing data) is
\begin{equation}
\mu_{\alpha}(\mathcal{A})=\frac{1}{T_s(\alpha|\varphi(\mathcal{A}))}=\frac{\alpha\mu}{H_{\varphi(\mathcal{A})}-H_{\varphi(\mathcal{A})-\alpha}}.
    \label{eq:SetServiceRate}
\end{equation}
It is not hard to see that 
\begin{equation}
\mu \varphi(\mathcal{A})\ge \mu_s(\alpha|\varphi(\mathcal{A})) \ge \mu (\varphi(\mathcal{A})-\alpha+1)
\label{ieq:rate}
\end{equation}
\ul{Shifted Exponential Service:}
In this model \cite{joshi2012coding}, delivery consists of  two steps: first, the node takes an exponential random time to process the request; second, the node takes a constant time proportional to the number of blocks to deliver them to the user. Therefore, the two step delivery time for a node storing $1/\alpha$ fraction of the file can be modeled by the shifted exponential distribution with rate $\mu$ and the shift parameter $\Delta/\alpha$.
For this model, we have
\begin{equation}
  T_s(\alpha|\varphi(\mathcal{A}))=\frac{\Delta}{\alpha}+\frac{1}{\mu}(H_{\varphi(\mathcal{A})}-H_{\varphi(\mathcal{A})-\alpha}),
\label{eq:shifttime}  
\end{equation}
where $\frac{\Delta}{\alpha}$ comes from the service rate shifting discussed above. The corresponding service rate from set $\mathcal{A}$ is
\begin{equation}
\mu_{\alpha}(\mathcal{A})=\frac{\alpha\mu}{\Delta\mu+\alpha(H_{\varphi(\mathcal{A})}-H_{\varphi(\mathcal{A})-\alpha})}.
    \label{eq:SetServiceRateshift}
\end{equation}
As $\varphi(\mathcal{A}) \in [\alpha,\alpha m]$, it is not hard to see that 
\begin{equation}
\frac{\mu \varphi(\mathcal{A})}{\Delta \mu +\alpha}\ge \mu_s(\alpha|\varphi(\mathcal{A})) \ge \frac{\alpha\mu (\varphi(\mathcal{A})-\alpha+1)}{\Delta \mu(\alpha m-\alpha+1)+\alpha^2}
\label{ieq:rateshift}
\end{equation}
\ul{DSS Service Rate:}
Under an $\alpha$-allocation, the DSS service rate is given by
\begin{equation}
\mu_s(\alpha) = 
\sum_{\mathcal{A}:\; \sum_{i \in \mathcal{A}} s_i \ge k}
P(\mathcal{A})\mu_{\alpha}(\mathcal{A})
\label{eq:service}
\end{equation}
where $\mu_{\alpha}(\mathcal{A})$ is the service rate when the set of accessed nodes is $\mathcal{A}$, given by \eqref{eq:SetServiceRate} or \eqref{eq:SetServiceRateshift},
and $P(\mathcal{A})$ is the probability of accessing set $\mathcal{A}$, given by \eqref{eq:prob}. 
\subsection{Preview of the Results and Future Work}
We argue that finding $\alpha$ that maximizes \eqref{eq:service} is hard.
We prove that $\mu_s(\alpha)$ is not always maximal for $\alpha=1$, and we specify two system parameter regions where 1) $\mu_s(\alpha)<\mu_s(1)$ and 2) $\mu_s(\alpha)>\mu_s(1)$. We numerically analyze the optimal storage allocation. We find that performance metrics $\mu_s(\alpha)$ (the service rate) and $P_s(\alpha)$ (the probability of successful recovery) may exhibit different trends with changing allocations.
We make conjectures on how optimal storage allocation changes with the parameter $r$, $p$, and $m$, which we will try to prove in future work.

\section{SYSTEM PERFORMANCE ANALYSIS}
\label{Sec:analysis}
\subsection{Fixed-size Access and Scaled Exponential Service}
\begin{claim}
\label{cor:fixed}
For fixed-size access model under scaled exponential distribution,
the DSS service rate \eqref{eq:service} becomes
\[
{\small
    \mu_s(\alpha)=\frac{\mu\alpha}{\binom N r} \sum^{\min(r,\alpha m)}_{\varphi=\alpha}\frac{1}{H_\varphi-H_{\varphi-\alpha}} \binom {\alpha m}{\varphi} \binom {N-\alpha m}{r-\varphi}}
\]
    \label{eq:scaled}
\end{claim}
The claim follows from the assertions in Sec.~\ref{Sec:model}.
We see that finding $\alpha$ that maximizes the $\mu_s(\alpha)$ is hard.
Instead, we prove below that $\alpha=1$ is not always optimal.
\begin{theorem}
 For the fixed-size access model, when the waiting time of each node follows scaled exponential distribution, the optimal $\mu_s(\alpha)$ isn't always reached at $\alpha=1$.
 \label{Le:fixed}
\end{theorem}

\begin{proof}
To prove $\alpha=1$ is not the optimal choice for all $r$, we will show that 1) there is a region of $r$ values s.t.\ $\mu_s(\alpha)<\mu_s(1)$ and 2) there is a region of $r$ values s.t.\ $\mu_s(\alpha)>\mu_s(1)$ for $\alpha>1$.
\\[1ex]
1) We consider $\mu_s(\alpha)<\mu_s(1)$ as follows:
According to \eqref{ieq:rate},
\begin{small}
\begin{align*}
    \mu_s(\alpha)\!&\!<\frac{\mu}{\binom{N}{r}}\sum^{\min(r,\alpha m)}_{\varphi=\alpha}\varphi \binom{\alpha m}{\varphi}\binom{N-\alpha m}{r-\varphi}\\
    &=\frac{\mu\alpha m}{\binom{N}{r}}\sum^{\min(r,\alpha m)}_{\varphi=\alpha} \binom{\alpha m-1}{\varphi-1}\binom{N-\alpha m}{r-\varphi}\\
    &=\frac{\mu\alpha m}{\binom{N}{r}}\!\sum^{\min(r,\alpha m)}_{\varphi=\alpha}\! \prod^{\alpha-2}_{i=0}\frac{\alpha m-1-i}{\varphi-1-i} \binom{\alpha m-\alpha}{\varphi-\alpha}\binom{N-\alpha m}{r-\varphi}
    \end{align*}
Since\ $\varphi$ goes from $\alpha$ to $\alpha m$, we further have 
    \begin{align*}
    \mu_s(\alpha)
    &<\frac{\mu\alpha m}{\binom{N}{r}}\sum^{\min(r,\alpha m)}_{\varphi=\alpha}(\prod^{\alpha-2}_{i=0}\frac{\alpha m-1-i}{\alpha-1-i}) \\ &\binom{\alpha m-\alpha}{\varphi-\alpha}\binom{N-\alpha m}{r-\varphi}\\
    &=\frac{\mu\alpha m\binom{\alpha m-1}{\alpha-1}}{\binom{N}{r}}\sum^{\min(r,\alpha m)}_{\varphi=\alpha} \binom{\alpha m-\alpha}{\varphi-\alpha}\binom{N-\alpha m}{r-\varphi}\\
    &=\frac{\mu\alpha m\binom{\alpha m-1}{\alpha-1} \binom{N-\alpha}{r-\alpha}}{\binom{N}{r}} ~~~~
    \text{(by Vandermonde's\ convolution)}
\end{align*}
\end{small}
As we know,
\begin{small}
\begin{align*}
    \mu_s(1)&=\sum^{m}_{\varphi=1}\mu_s(1|\varphi)\frac{\binom{m}{\varphi}\binom{N-m}{r-\varphi}}{\binom{N}{r}}\\
    &=\frac{\mu}{\binom{N}{r}}\sum^{m}_{\varphi=1}\varphi\binom{m}{\varphi}\binom{N-m}{r-\varphi}=\frac{\mu m\binom{N-1}{r-1}}{\binom{N}{r}}
\end{align*}
\end{small}
Then to satisfy $\mu_s(\alpha)<\mu_s(1)$, we need
\begin{equation}
\begin{split}
   &\frac{\mu\alpha m\binom{\alpha m-1}{\alpha-1} \binom{N-\alpha}{r-\alpha}}{\binom{N}{r}}<\frac{\mu m\binom{N-1}{r-1}}{\binom{N}{r}}\\
   &\Leftrightarrow \alpha \binom{\alpha m-1}{\alpha-1}<\prod^{\alpha-2}_{i=0}\frac{N-1-i}{r-1-i}
\end{split}
\label{eq:shift-fixb}
\end{equation}
As $\frac{N-1-i}{r-1-i}<\frac{N-2-i}{r-2-i}$ for $N>r$, we have
$\prod^{\alpha-2}_{i=0}\frac{N-1-i}{r-1-i}>(\frac{N-1}{r-1})^{\alpha-1}$. Inequality \eqref{eq:shift-fixb} is true when
\begin{small}
\begin{align*}
    &\alpha \binom{\alpha m-1}{\alpha-1}<(\frac{N-1}{r-1})^{\alpha-1}
    \Leftrightarrow r<1+\frac{N-1}{\sqrt[\alpha-1]{\alpha \binom{\alpha m-1}{\alpha-1}}}
\end{align*}
\end{small}
\noindent Thus $\mu_s(\alpha)<\mu_s(1)$ for $r\in \Big[\alpha,1+\frac{N-1}{\sqrt[\alpha-1]{\alpha \binom{\alpha m-1}{\alpha-1}}}\Big)$.
\\[1ex]
2) We consider $\mu_s(\alpha)>\mu_s(1)$ as follows:
According to \eqref{ieq:rate},
\begin{small}
\begin{align*}
    \mu_s(\alpha)&>\frac{\mu}{\binom{N}{r}}\sum^{\min(r,\alpha m)}_{\varphi=\alpha}(\varphi-\alpha+1) \binom{\alpha m}{\varphi}\binom{N-\alpha m}{r-\varphi}\\
    &=\frac{\mu}{\binom{N}{r}}\sum^{\min(r,\alpha m)}_{\varphi=\alpha}(\varphi-\alpha+1) \prod^{\alpha-2}_{i=0}\frac{\alpha m-i}{\varphi-i}\\ &\binom{\alpha m-\alpha+1}{\varphi-\alpha+1}\binom{N-\alpha m}{r-\varphi}
    \end{align*}
 Since\ $\varphi$ goes from $\alpha$ to $\alpha m$, we further have 
    \begin{align*}
\mu_s(\alpha)
    &>\frac{\mu}{\binom{N}{r}}\sum^{\min(r,\alpha m)}_{\varphi=\alpha}(\varphi-\alpha+1)\binom{\alpha m-\alpha+1}{\varphi-\alpha+1}\binom{N-\alpha m}{r-\varphi}\\
    &=\frac{\mu(\alpha m-\alpha+1)\binom{N-\alpha}{r-\alpha}}{\binom{N}{r}}~~~
    \text{(Vandermonde's convolution)}
\end{align*}
\end{small}
As we know, $\mu_s(1)=\frac{\mu m\binom{N-1}{r-1}}{\binom{N}{r}}$
\\
Then to satisfy $\mu_s(\alpha)>\mu_s(1)$, we need
\begin{equation}
\begin{split}
   &\frac{\mu(\alpha m-\alpha+1)\binom{N-\alpha}{r-\alpha}}{\binom{N}{r}}>\frac{\mu m\binom{N-1}{r-1}}{\binom{N}{r}}\\
   &\Leftrightarrow \frac{\alpha m-\alpha +1}{m}>\prod^{\alpha-2}_{i=0}\frac{N-1-i}{r-1-i}
\end{split}
\label{eq:scaled-fix2}
\end{equation}
As $\frac{N-1-i}{r-1-i}<\frac{N-2-i}{r-2-i}$ for $N>r$, we have
$\prod^{\alpha-2}_{i=0}\frac{N-1-i}{r-1-i}<(\frac{N-\alpha+1}{r-\alpha+1})^{\alpha-1}$. Inequality \eqref{eq:scaled-fix2} is true when
\begin{small}
\begin{align*}
    &\frac{\alpha m-\alpha +1}{m}>(\frac{N-\alpha+1}{r-\alpha+1})^{\alpha-1}\\
    &\Leftrightarrow r>\sqrt[\alpha-1]{\frac{m}{\alpha m-\alpha +1}}(N-\alpha+1)+\alpha-1
\end{align*}
\end{small}

Thus $\mu_s(\alpha)>\mu_s(1)$  when $r\in \Big(\sqrt[\alpha-1]{\frac{m}{\alpha m-\alpha +1}}(N-\alpha+1)+\alpha-1,N\Big].$
\end{proof}

In the proof above, we find two regions of $r$ which can show when $\mu_s(1)$ reaches the maximum. Here we give some examples to analyze these two regions. Let's give the parameter set as $(N,m,\mu,\alpha)$. When the parameter set is (30,2,1,4), we can get two regions [4,6.5] for $\mu_s(\alpha)<\mu_s(1)$, and [22.9,30] for $\mu_s(\alpha)>\mu_s(1)$, and both regions are exist. But there is a gap between the two regions, which means when $r$ is in (6.5,22.9), we can not decide whether $\mu_s(1)$ is maximum or not. The big gap appears because of the bounds we used in proof are not tight enough. When the parameter set is (30,3,1,5), we can get another two regions [4,4.4] for $\mu_s(\alpha)<\mu_s(1)$ and [22.7,30] for $\mu_s(\alpha)>\mu_s(1)$, and the first region is not exist for $r\ge\alpha$. 

Although these two regions can not help us to make a decision under any parameter set, they can still tell us something more about the system model in Conjecture \ref{hyp:fix1}. 
\begin{conjecture}
When $r$ is small, $\mu_s(1)$ is more likely to be the maximum; When $r$ is large, the maximum $\mu_s(\alpha)$ is not at 1, and the optimal $\alpha$ is increasing with $r$.
\label{hyp:fix1}
\end{conjecture}

\subsection{Probabilistic Access and Scaled Exponential Service}
\begin{claim}
Under the probabilistic access model, 
\[
{\small
    \mu_s(\alpha)=\sum^{\alpha m}_{\varphi=\alpha}\frac{\mu\alpha}{H_\varphi-H_{\varphi-\alpha}} \binom {\alpha m}{\varphi} (1-p)^{\varphi}p^{\alpha m-\varphi}}
\]
\label{cor:prob}
\end{claim}
The claim also follows from the assertions in Sec.~\ref{Sec:model}.
We see that finding $\alpha$ that maximizes the $\mu_s(\alpha)$ is still hard.
Therefore, we prove below that $\alpha=1$ is not always optimal.

\begin{theorem}
 For the probabilistic access model, when the waiting time of each node follows scaled exponential distribution, the optimal $\mu_s(\alpha)$ isn't always reached at $\alpha=1$.
 \label{Le:prob}
\end{theorem}

\begin{proof}
To prove $\alpha=1$ is not the optimal choice for all $p$, we will show that 1) there is a region of $p$ values s.t.\ $\mu_s(\alpha)<\mu_s(1)$ and 2) there is a region of $p$ values s.t.\ $\mu_s(\alpha)>\mu_s(1)$ for $\alpha>1$.
\\[1ex]
1) We consider $\mu_s(\alpha)<\mu_s(1)$ as follows:
According to \eqref{ieq:rate},
\begin{small}
\begin{align*}
    \mu_s(\alpha)&<\mu\sum^{\alpha m}_{\varphi=\alpha}\varphi\binom{\alpha m}{\varphi}(1-p)^{\varphi}p^{\alpha m-\varphi}\\
    &=\mu\alpha m\sum^{\alpha m}_{\varphi=\alpha}\binom{\alpha m-1}{\varphi-1}(1-p)^{\varphi}p^{\alpha m-\varphi}\\
    &=\mu\alpha m\sum^{\alpha m}_{\varphi=\alpha}(\prod^{\alpha-2}_{i=0}\frac{\alpha m-1-i}{\varphi-1-i})\binom{\alpha m-\alpha}{\varphi-\alpha}(1-p)^{\varphi}p^{\alpha m-\varphi}
     \end{align*}
Since $\varphi$ goes from $\alpha$ to $\alpha m$, we have
    \begin{align*}
    \mu_s(\alpha)&<\mu\alpha m\binom{\alpha m-1}{\alpha-1}(1-p)^{\alpha}\sum^{\alpha m-\alpha}_{\varphi=0}\binom{\alpha m-\alpha}{\varphi}\\&(1-p)^{\varphi}p^{\alpha m-\alpha-\varphi}\\
    &By\ using\ binomial\ expansion\ ,\\
    &=\mu\alpha m\binom{\alpha m-1}{\alpha-1}(1-p)^{\alpha}
\end{align*}
\end{small}
As we know,
\begin{small}
\begin{align*}
    \mu_s(1)&=\sum^{m}_{\varphi=1}\mu_s(1|\varphi)\binom{m}{\varphi}(1-p)^{\varphi}p^{m-\varphi}\\
    &=\mu\sum^{m}_{\varphi=1}\varphi\binom{m}{\varphi}(1-p)^{\varphi}p^{m-\varphi}
    =\mu m(1-p)
\end{align*}
\end{small}
Then to satisfy $\mu_s(\alpha)<\mu_s(1)$, we need
\begin{small}
\begin{align*}
   &\mu\alpha m\binom{\alpha m-1}{\alpha-1}(1-p)^{\alpha}<\mu m(1-p)\\
   &\Leftrightarrow (1-p)^{\alpha-1}<\frac{1}{\alpha\binom{\alpha m-1}{\alpha-1}}
\Leftrightarrow p>1-\frac{1}{\sqrt[\alpha-1]{\alpha\binom{\alpha m-1}{\alpha-1}}}
\end{align*}
\end{small}
Thus $\mu_s(\alpha)<\mu_s(1)$ for
$p\in \Big(1-\frac{1}{\sqrt[\alpha-1]{\alpha\binom{\alpha m-1}{\alpha-1}}},1\Big]$.
\\[1ex]
2) We consider $\mu_s(\alpha)>\mu_s(1)$ as follows:
\\[1ex]
According to \eqref{ieq:rate},
\begin{small}
\begin{align*}
    \mu_s(\alpha)&>\mu\sum^{\alpha m}_{\varphi=\alpha}(\varphi-\alpha+1)\binom{\alpha m}{\varphi}(1-p)^{\varphi}p^{\alpha m-\varphi}\\
    &=\mu\sum^{\alpha m}_{\varphi=\alpha}(\varphi-\alpha+1)(\prod^{\alpha-2}_{i=0}\frac{\alpha m-i}{\varphi-i})\binom{\alpha m-\alpha+1}{\varphi-\alpha+1}\\&(1-p)^{\varphi}p^{\alpha m-\varphi}     \end{align*}
Since $\varphi$ goes from $\alpha$ to $\alpha m$, we have
    \begin{align*}
    \mu_s(\alpha)&>\mu\sum^{\alpha m}_{\varphi=\alpha}(\varphi-\alpha+1)\binom{\alpha m-\alpha+1}{\varphi-\alpha+1}(1-p)^{\varphi}p^{\alpha m-\varphi}\\
    &=\mu(\alpha m-\alpha+1)(1-p)^{\alpha}\sum^{\alpha m-\alpha}_{\varphi=0}\binom{\alpha m-\alpha}{\varphi}\\&(1-p)^{\varphi}p^{\alpha m-\alpha-\varphi}
    =\mu(\alpha m-\alpha+1)(1-p)^{\alpha}
\end{align*}
\end{small}
As we know $\mu_s(1)=\mu m(1-p)$
Then to satisfy $\mu_s(\alpha)>\mu_s(1)$, we need
\begin{small}
\begin{align*}
   &\mu(\alpha m-\alpha+1)(1-p)^{\alpha}>\mu m(1-p)\\
   &\Leftrightarrow (1-p)^{\alpha-1}>\frac{m}{\alpha m-\alpha+1}
   \Leftrightarrow p<1-\sqrt[\alpha-1]{\frac{m}{\alpha m-\alpha+1}}
\end{align*}
\end{small}
Thus $\mu_s(\alpha)>\mu_s(1)$ for $p\in \Big[0,1-\sqrt[\alpha-1]{\frac{m}{\alpha m-\alpha+1}}\Big)$.
\end{proof}

In the proof above, we find two regions of $p$ which can show when $\mu_s(1)$ reaches the maximum. It is easy to see that both regions are exist for any parameter set. Here we can also give a parameter set as $(m,\mu,\alpha)$ to show an example of the two regions in the proof. When the parameter set is (2,1,4), we can get two regions [0.81,1] for $\mu_s(\alpha)<\mu_s(1)$, and [0,0.26] for $\mu_s(\alpha)>\mu_s(1)$. There is also a gap between two regions which shows the undecided region.

Here we give the Conjecture \ref{hyp:prob1} according to the $p$'s regions. 
\begin{conjecture}
When $p$ is large, $\mu_s(1)$ is more likely to be the maximum; When $p$ is small, the maximum $\mu_s(\alpha)$ is not at 1, and the optimal $\alpha$ is decreasing with $p$.
\label{hyp:prob1}
\end{conjecture}

\subsection{Fixed-size Access and Shifted Exponential Service}
\begin{claim}
\label{cor:fixeds}
For fixed-size access model under shifted exponential distribution,
the DSS servise rate $\mu_s(\alpha)$ is given by
\[
{\small
    \frac{\mu\alpha}{\binom N r} \sum^{\min(r,\alpha m)}_{\varphi=\alpha}\frac{1}{\Delta\mu+\alpha(H_{\varphi}-H_{\varphi-\alpha})} \binom {\alpha m}{\varphi} \binom {N-\alpha m}{r-\varphi}}
\]
\label{eq:shift}
\end{claim}
The claim follows from the assertions in Sec.~\ref{Sec:model}.
Similarly, finding $\alpha$ that maximizes the $\mu_s(\alpha)$ is hard.
Instead, we prove below that $\alpha=1$ is not always optimal.

\begin{theorem}
 For the fixed-size access model, when the waiting time of each node follows shifted exponential distribution, the optimal $\mu_s(\alpha)$ isn't always reached at $\alpha=1$.
 \label{Le:shift-fix}
\end{theorem}

\begin{proof}
This proof is similar as which for the Theorem \ref{Le:fixed}. Therefore we only keep some key steps.

1) We consider $\mu_s(\alpha)<\mu_s(1)$ as follows:
According to \eqref{ieq:rateshift},
\begin{small}
\begin{align*}
    \mu_s(\alpha)&<\frac{\mu}{(\Delta\mu+\alpha)\binom{N}{r}}\sum^{\min(r,\alpha m)}_{\varphi=\alpha}\varphi \binom{\alpha m}{\varphi}\binom{N-\alpha m}{r-\varphi}\\
    &<\frac{\mu\alpha m\binom{\alpha m-1}{\alpha-1} \binom{N-\alpha}{r-\alpha}}{(\Delta\mu+\alpha)\binom{N}{r}}
\end{align*}
\end{small}
As we know,
\begin{small}
\begin{align*}
    \mu_s(1)
    &>\frac{\mu}{(\Delta\mu m+1)\binom{N}{r}}\sum^{m}_{\varphi=1}\varphi\binom{m}{\varphi}\binom{N-m}{r-\varphi}
    =\frac{\mu m\binom{N-1}{r-1}}{(\Delta\mu m+1)\binom{N}{r}}
\end{align*}
\end{small}
Then to satisfy $\mu_s(\alpha)<\mu_s(1)$, we need
\begin{equation}
\begin{split}
   &\frac{\mu\alpha m\binom{\alpha m-1}{\alpha-1}\binom{N-\alpha}{r-\alpha} }{(\Delta\mu+\alpha)\binom{N}{r}}<\frac{\mu m\binom{N-1}{r-1}}{(\Delta\mu m+1)\binom{N}{r}}\\
   &\Leftrightarrow \frac{\alpha(\Delta\mu m+1)\binom{\alpha m-1}{\alpha-1}}{\Delta\mu+\alpha}<\prod^{\alpha-2}_{i=0}\frac{N-1-i}{r-1-i}
\end{split}
\label{eq:shift-fix}
\end{equation}
As $\frac{N-1-i}{r-1-i}<\frac{N-2-i}{r-2-i}$ for $N>r$, we have
$\prod^{\alpha-2}_{i=0}\frac{N-1-i}{r-1-i}>(\frac{N-1}{r-1})^{\alpha-1}$. Inequality \eqref{eq:shift-fix} is true when
\begin{small}
\begin{align*}
    &\frac{\alpha(\Delta\mu m+1)\binom{\alpha m-1}{\alpha-1}}{\Delta\mu+\alpha}<(\frac{N-1}{r-1})^{\alpha-1}\\
    &\Leftrightarrow r<1+\sqrt[\alpha-1]{\frac{\Delta\mu+\alpha}{\alpha(\Delta\mu m+1)\binom{\alpha m-1}{\alpha-1}}}(N-1)
\end{align*}
\end{small}
Therefore, for $r\in \Big[\alpha,1+\sqrt[\alpha-1]{\frac{\Delta\mu+\alpha}{\alpha(\Delta\mu m+1)\binom{\alpha m-1}{\alpha-1}}}(N-1)\Big)$, $\mu_s(\alpha)<\mu_s(1)$ is true.
\\[1ex]
2) We consider $\mu_s(\alpha)>\mu_s(1)$ as follows:
\\[1ex]
According to \eqref{ieq:rateshift},
\begin{small}
\begin{align*}
    \mu_s(\alpha)&>\frac{\mu\alpha}{(\Delta\mu(\alpha m-\alpha+1)+\alpha^2)\binom{N}{r}}\sum^{\min(r,\alpha m)}_{\varphi=\alpha}(\varphi-\alpha+1)\\ &\binom{\alpha m}{\varphi}\binom{N-\alpha m}{r-\varphi}\\
    &>\frac{\mu\alpha(\alpha m-\alpha+1)\binom{N-\alpha}{r-\alpha}}{(\Delta\mu(\alpha m-\alpha+1)+\alpha^2)\binom{N}{r}}
\end{align*}
\end{small}
As we know,
\begin{small}
\begin{align*}
    \mu_s(1)
    &<\frac{\mu}{(\Delta\mu +1)\binom{N}{r}}\sum^{m}_{\varphi=1}\varphi\binom{m}{\varphi}\binom{N-m}{r-\varphi}=\frac{\mu m\binom{N-1}{r-1}}{(\Delta\mu +1)\binom{N}{r}}
\end{align*}
\end{small}
Then to satisfy $\mu_s(\alpha)>\mu_s(1)$, we need
\begin{equation}
\begin{split}
   &\frac{\mu\alpha(\alpha m-\alpha+1)\binom{N-\alpha}{r-\alpha}}{(\Delta\mu(\alpha m-\alpha+1)+\alpha^2)\binom{N}{r}}>\frac{\mu m\binom{N-1}{r-1}}{(\Delta\mu +1)\binom{N}{r}}\\
   &\Leftrightarrow \frac{\alpha(\Delta\mu+1)(\alpha m-\alpha +1)}{m(\Delta\mu(\alpha m-\alpha+1)+\alpha^2)}>\prod^{\alpha-2}_{i=0}\frac{N-1-i}{r-1-i}
\end{split}
\label{eq:shift-fix2}
\end{equation}
As $\frac{N-1-i}{r-1-i}<\frac{N-2-i}{r-2-i}$ for $N>r$, we have
$\prod^{\alpha-2}_{i=0}\frac{N-1-i}{r-1-i}<(\frac{N-\alpha+1}{r-\alpha+1})^{\alpha-1}$. Inequality \eqref{eq:shift-fix2} is true when

\begin{small}
\begin{align*}
    &\frac{\alpha(\Delta\mu+1)(\alpha m-\alpha +1)}{m(\Delta\mu(\alpha m-\alpha+1)+\alpha^2)}>(\frac{N-\alpha+1}{r-\alpha+1})^{\alpha-1}\\
    &\Leftrightarrow r>\sqrt[\alpha-1]{\frac{\Delta\mu m(\alpha m-\alpha+1)+\alpha^2 m}{\alpha(\Delta\mu+1)(\alpha m-\alpha+1)}}\\ &(N-\alpha+1)+\alpha-1
\end{align*}
\end{small}
Therefore, for $r\in \Big(\sqrt[\alpha-1]{\frac{\Delta\mu m(\alpha m-\alpha+1)+\alpha^2 m}{\alpha(\Delta\mu+1)(\alpha m-\alpha+1)}}(N-\alpha+1)+\alpha-1,N\Big]$, $\mu_s(\alpha)>\mu_s(1)$ is true.
\end{proof}

In the proof above, we find two regions of $r$, which are similar as what we find in Theorem \ref{Le:fixed}, can show when $\mu_s(1)$ is the maximum . Here we can give a parameter set as $(N,m,\mu,\Delta,\alpha)$ to show an example of the two regions in the proof. When the parameter set is (30,2,1,10,4), we can get two regions [4,5.8] for $\mu_s(\alpha)<\mu_s(1)$, and [25.7,30] for $\mu_s(\alpha)>\mu_s(1)$, and both regions are exist. When the parameter set is (30,3,1,10,6), we can get another two regions [6,4.1] for $\mu_s(\alpha)<\mu_s(1)$ and [27.4,30] for $\mu_s(\alpha)>\mu_s(1)$, and the first region is not exist.  When the parameter set is (30,2,1,1,4), we can get another two regions [4,7.6] for $\mu_s(\alpha)<\mu_s(1)$ and [30.4,30] for $\mu_s(\alpha)>\mu_s(1)$, and the second region is not exist. 
Here the Conjecture \ref{hyp:fix1} holds. 

\subsection{Probabilistic Access and Shifted Exponential Service}
\begin{claim}
Under the probabilistic access model, 
\[
{\small
    \mu_s(\alpha)\!=\!\sum^{\alpha m}_{\varphi=\alpha}\frac{\mu\alpha}{\Delta\mu+\alpha(H_{\varphi}-H_{\varphi-\alpha})} \binom {\alpha m}{\varphi} (1-p)^{\varphi}p^{\alpha m-\varphi}}
\]
\label{cor:probs}
\end{claim}
The claim follows from the assertions in Sec.~\ref{Sec:model}, and finding $\alpha$ that maximizes the $\mu_s(\alpha)$ remains hard.
Therefore, we prove below that $\alpha=1$ is not always optimal.
\begin{theorem}
 For the fixed-size access model, when the waiting time of each node follows shifted exponential distribution, the optimal $\mu_s(\alpha)$ isn't always reached at $\alpha=1$.
 \label{Le:shift-prob}
\end{theorem}

\begin{proof}
This proof is similar as which for the Theorem \ref{Le:prob}. Therefore we only keep some key steps.

1) We consider $\mu_s(\alpha)<\mu_s(1)$ as follows:
According to \eqref{ieq:rateshift},
\begin{small}
\begin{align*}
    \mu_s(\alpha)&<\frac{\mu}{\Delta\mu+\alpha}\sum^{\alpha m}_{\varphi=\alpha}\varphi\binom{\alpha m}{\varphi}(1-p)^{\varphi}p^{\alpha m-\varphi}\\
    &<\frac{\mu\alpha m\binom{\alpha m-1}{\alpha-1}(1-p)^{\alpha}}{\Delta\mu+\alpha}
\end{align*}
\end{small}
As we know,
\begin{small}
\begin{align*}
    \mu_s(1)
    &>\frac{\mu}{\Delta\mu m+1}\sum^{m}_{\varphi=1}\varphi\binom{m}{\varphi}(1-p)^{\varphi}p^{m-\varphi}=\frac{\mu m(1-p)}{\Delta\mu m+1}
\end{align*}
\end{small}
Then to satisfy $\mu_s(\alpha)<\mu_s(1)$, we need
\begin{small}
\begin{align*}
   &\frac{\mu\alpha m\binom{\alpha m-1}{\alpha-1}(1-p)^{\alpha}}{\Delta\mu+\alpha}<\frac{\mu m(1-p)}{\Delta\mu m+1}\\
   &\Leftrightarrow p>1-\sqrt[\alpha-1]{\frac{\Delta\mu+\alpha}{\alpha(\Delta\mu m+1)\binom{\alpha m-1}{\alpha-1}}}
\end{align*}
\end{small}
$p\in \Big(1-\sqrt[\alpha-1]{\frac{\Delta\mu+\alpha}{\alpha(\Delta\mu m+1)\binom{\alpha m-1}{\alpha-1}}},1\Big]$, $\mu_s(\alpha)<\mu_s(1)$ is true.
\\[1ex]
2) We consider $\mu_s(\alpha)>\mu_s(1)$ as follows:
According to \eqref{ieq:rateshift},
\begin{small}
\begin{align*}
    \mu_s(\alpha)&>\frac{\alpha\mu}{\Delta\mu(\alpha m-\alpha+1)+\alpha^2}\sum^{\alpha m}_{\varphi=\alpha}(\varphi-\alpha+1)\binom{\alpha m}{\varphi}\\&(1-p)^{\varphi}p^{\alpha m-\varphi}\\
    &>\frac{\alpha\mu(\alpha m-\alpha+1)(1-p)^{\alpha}}{\Delta\mu(\alpha m-\alpha+1)+\alpha^2}
\end{align*}
\end{small}
As we know,
\begin{small}
\begin{align*}
    \mu_s(1)
    &<\frac{\mu}{\Delta\mu +1}\sum^{m}_{\varphi=1}\varphi\binom{m}{\varphi}(1-p)^{\varphi}p^{m-\varphi}
    =\frac{\mu m(1-p)}{\Delta\mu +1}
\end{align*}
\end{small}
Then to satisfy $\mu_s(\alpha)>\mu_s(1)$, we need
\begin{small}
\begin{align*}
   &\frac{\alpha\mu(\alpha m-\alpha+1)(1-p)^{\alpha}}{\Delta\mu(\alpha m-\alpha+1)+\alpha^2}>\frac{\mu m(1-p)}{\Delta\mu +1}\\
   &\Leftrightarrow p<1-\sqrt[\alpha-1]{\frac{m(\Delta\mu(\alpha m-\alpha+1)+\alpha^2)}{\alpha(\Delta\mu +1)(\alpha m-\alpha+1)}}
\end{align*}
\end{small}
 Therefore, for $p\in \Big[0,1-\sqrt[\alpha-1]{\frac{m(\Delta\mu(\alpha m-\alpha+1)+\alpha^2)}{\alpha(\Delta\mu +1)(\alpha m-\alpha+1)}}\Big)$, $\mu_s(\alpha)>\mu_s(1)$ is true.
\end{proof}

From the proof above, we find two regions of $p$, similar as what we found in Theorem \ref{Le:prob}, and can see when $\mu_s(1)$ is the maximum. Here we give some examples to analyze these two regions. We also give a parameter set as $(m,\mu,\Delta,\alpha)$ to show an example of the two regions in the proof. When the parameter set is (2,1,10,4), we can get two regions [0.83,1] for $\mu_s(\alpha)<\mu_s(1)$,and [0,0.15] for $\mu_s(\alpha)>\mu_s(1)$, and both regions are exist. When the parameter set is (2,1,10,4), we can get another two regions [0,-0.016] for $\mu_s(\alpha)<\mu_s(1)$ and [0.77,1] for $\mu_s(\alpha)>\mu_s(1)$, and the first region is not exist. 
Here the Conjecture \ref{hyp:prob1} holds. 

From Theorems \ref{Le:fixed}, \ref{Le:prob}, \ref{Le:shift-fix} and \ref{Le:shift-prob}, we see that the optimal allocation varies in accordance with different system parameters. Recall that $\alpha=1$ was found to be universally optimal in \cite{noori2016storage} where it was assumed that the download time does not scale with the size of data.

\section{OPTIMAL STORAGE ALLOCATION ANALYSIS}
\label{Sec:simulation}
We next numerically analyze the optimal storage allocation. We compute the service rate and probability of successful recovery with the allocation parameter $\alpha$. Since the accessed nodes number $r$, coded file size ratio $m$ and failure probability $p$ are the key parameters for the the storage system, we also vary these values to see how the optimal allocation changes.

According to the formulas in Claim \ref{cor:fixed} and \ref{cor:prob}, we know that the rate parameter $\mu$ in scaled exponential distribution doesn't affect the numerical analysis results. But from Claim \ref{cor:fixeds} and \ref{cor:probs}, we can see both rate parameter $\mu$ and shift parameter $\Delta$ affect the numerical analysis results. When $\Delta\mu \ll 1$, the shifted exponential distribution is equivalent to an exponential distribution, then the minimum spreading allocation is universally optimal; When $\Delta\mu \gg 1$, the shifted exponential distribution is equivalent to a constant, then the $\mu_s(\alpha)$ is changing with the probability of successful recovery. Therefore, we select the $\mu=1$ and $\Delta=3$ in the simulations below as appropriate values.

\subsection{Fixed-size Access}
For fixed-sized access model, we present two figures to analyze the optimal storage allocation in the interval $\alpha\in[1,5]$ or $[1,6]$. In Fig.~\ref{fig:Fixed1}, we have three subfigures, the left is the average service rate for the scaled exponential distribution, the middle is for the shifted exponential distribution, and the right is the probability of successful recovery. Firstly, let's analyze the left and right subfigures. When $m=3$ and $4$,  $\mu_s(\alpha)$ and $P_s(\alpha)$ are decreasing with $\alpha$, and the optimal allocation is $\alpha=1$. When $m=5$, the largest $\mu_s(\alpha)$ is reached at $\alpha=3$, but the $P_s(\alpha)$ is decreasing, then it is better to select $\alpha$ between $1$ to $3$ based on the weight of $\mu_s(\alpha)$ and $P_s(\alpha)$; When $m=6$, $\mu_s(\alpha)$ is increasing and $P_s(\alpha)$ reaches maximum at $\alpha=5$, then the optimal allocation is $\alpha=5$. Secondly, the middle subfigure has a similar pattern as the left one. And the only different is when $m=4$, the largest $\mu_s(\alpha)$ is reached at $\alpha=2$.
\begin{figure}[t]
    \centering
    \includegraphics[width=0.48\textwidth]{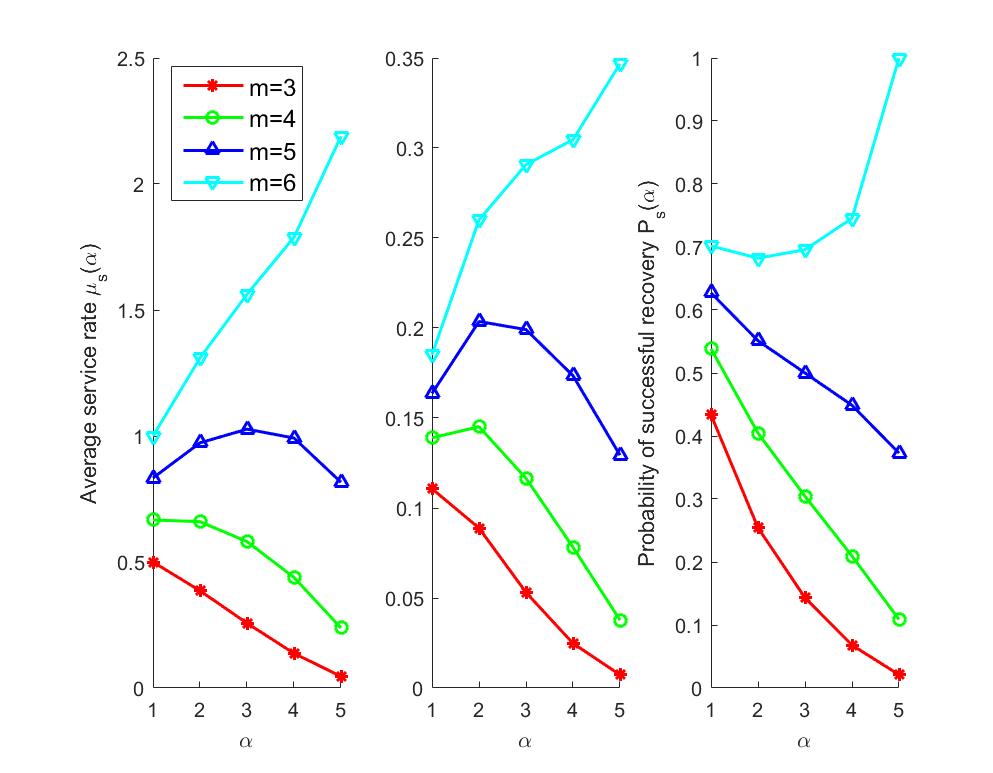}
    \caption{The average service rate and probability of successful recovery of fixed-size access model are changing with the allocation parameter $\alpha$ under different $m$ when $N=30$, $r=5$. Left: the distribution is scaled exponential with $\mu=1$; Middle: the distribution is shifted exponential with $\mu=1$ and $\Delta=3$; Right: the probability of successful recovery.}
    \label{fig:Fixed1}
\end{figure}

Fig.~\ref{fig:Fixed2} shows similar results. Firstly,let's analyze the left and right subfigures. When $r=6$, both $\mu_s(\alpha)$ and $P_s(\alpha)$ are decreasing, the optimal allocation is $\alpha=1$. When $r=7$ and $8$, $P_s(\alpha)$ is still decreasing, but $\mu_s(\alpha)$ reaches maximum at $\alpha=2$ and $3$, then the optimal allocation $\alpha$ is between $1$ to $2$ or $1$ to $3$. When $r=9$, $\mu_s(\alpha)$ is increasing, but $P_s(\alpha)$ is decreasing, then optimal allocation is selected based on the weight of $\mu_s(\alpha)$ and $P_s(\alpha)$. Secondly, the middle subfigure shows similar results as left subfigure except when $r=9$, the largest $\mu_s(\alpha)$ is reached at $\alpha=3$.

\begin{figure}[t]
    \centering
    \includegraphics[width=0.48\textwidth]{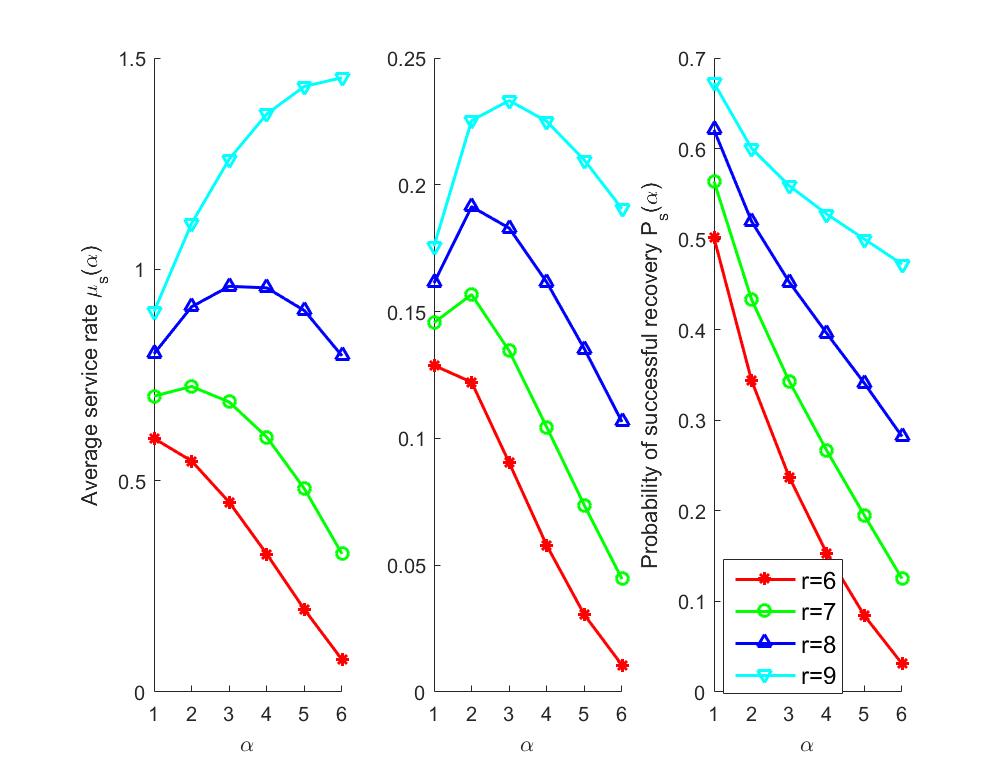}
    \caption{The average service rate and probability of successful recovery of fixed-size access model are changing with the allocation parameter $\alpha$ under different $r$ when $N=30$, $m=3$. Left: the distribution is scaled exponential with $\mu=1$; Middle: the distribution is shifted exponential with $\mu=1$ and $\Delta=3$; Right: the probability of successful recovery.}
    \label{fig:Fixed2}
\end{figure}

We can conclude that the optimal allocation is not fixed at $\alpha=1$ as we proved in the previous section, and here we can see that the optimal $\alpha$ is changing with different system parameters. The following claim helps to understand why the optimal $\alpha$ varies.

\begin{claim}
For fixed-size access model with $r=N$, if we use scaled exponential distribution, we have $\mu_s(\alpha)$<$\mu_s(\alpha+1)$; If we use shifted exponential distribution, the result varies based on different parameters' values.
\label{explainfix}
\end{claim}

\begin{proof}
1) We consider scaled exponential distribution:
\\[1ex]
Note that under $\alpha$-allocation, exactly $\alpha m$ nodes contain data and under $(\alpha+1)$-allocation, exactly $(\alpha+1)m$ nodes contain data. We will show that $\mu_s(\alpha)$<$\mu_s(\alpha+1)$ by showing that
$T_s(\alpha|\alpha m)>T_s(\alpha+1|(\alpha+1)m)$:
\begin{small}
\begin{align*}
T_s(\alpha|\alpha m)&=\frac{1}{\alpha \mu}(H_{\alpha m}-H_{\alpha m-\alpha})~~~~\text{cf.~\eqref{eq:time}} \\
&=\frac{1}{(1+\alpha)\mu}\Big(1+\frac{1}{\alpha}\Big)\sum_{i=1}^{\alpha }\frac{1}{\alpha m-\alpha +i}\\
&>\frac{1}{(1+\alpha)\mu}\Bigg(\sum_{i=1}^{\alpha }\frac{1}{\alpha m-\alpha +i}+\frac{1}{\alpha m}\Bigg)\\
&>\frac{1}{(1+\alpha)\mu}\Bigg(\sum_{i=0}^{\alpha}\frac{1}{\alpha m +m-\alpha+i}\Bigg)\\
&=\frac{1}{(1+\alpha)\mu}\big(H_{\alpha m+m}-H_{(\alpha+1)(m-1)}\big)\\[1pt]
&=T_s(\alpha+1|(\alpha+1)m)
\end{align*}
\end{small}
2) We consider shifted exponential distribution:
\\[1ex]
Similarly, we can get $T_s(\alpha|\alpha m)$:
\begin{small}
\begin{align*}
    T_s(\alpha|\alpha m)=\frac{\Delta}{\alpha}+\frac{1}{\mu}(H_{\alpha m}-H_{\alpha m-\alpha})
\end{align*}
\end{small}
If $\Delta\gg 1/\mu$, we can say $T_s(\alpha|\alpha m)=\frac{\Delta}{\alpha}$, then $T_s(\alpha|\alpha m)>T_s(\alpha+1|(\alpha+1)m)$ is obvious. If $\Delta\ll 1/\mu$, we can say $T_s(\alpha|\alpha m)=\frac{1}{\mu}(H_{\alpha m}-H_{\alpha m-\alpha})$, when $m=1$, $T_s(\alpha|\alpha m)<T_s(\alpha+1|(\alpha+1)m)$ is obvious.
\end{proof}

From Claim~\ref{explainfix}, if we use scaled exponential distribution, the average service rate $\mu_s(\alpha)$ is increasing with $\alpha$ when the probability of successful recovery is $1$. Meanwhile from Figs.~\ref{fig:Fixed1} and \ref{fig:Fixed2}, we know the probability of success access $P_s(\alpha)$ is decreasing under some setups of storage system. Then the pattern of $\mu_s(\alpha)$ is decided by the impact of $P_s(\alpha)$. If $P_s(\alpha)$ has a higher impact, e.g.,  $m=3$ and $4$ or $r=6$ and $7$, $\mu_s(\alpha)$ is changing with the pattern of  $P_s(\alpha)$; If $P_s(\alpha)$ has a lower impact, e.g.  $m=6$ or $r=9$, $\mu_s(\alpha)$ is increasing with $\alpha$; If the impact of  $P_s(\alpha)$ is in between, e.g. $m=5$ or $r=8$, $\mu_s(\alpha)$ is like a parabola. 

If we use shifted exponential, the average service rate $\mu_s(\alpha)$ is changing according to different parameter values. From Figs.~\ref{fig:Fixed1} and \ref{fig:Fixed2}, we know when $\mu=1$ and $\Delta=3$, the pattern of $\mu_s(\alpha)$ is similar as scaled exponential distribution's results.

\subsection{Probabilistic Access}
For probabilistic access model, we present three figures to analyze the optimal storage allocation in the interval $\alpha\in[1,10]$. In Fig.~\ref{fig:Prob1}, we have three subfigures, the left is the average service rate for the scaled exponential service, the middle is for the shifted exponential service, and the right is the probability of successful recovery. When $m=1$, both $\mu_s(\alpha)$ and $P_s(\alpha)$ are decreasing, then the optimal allocation is $\alpha=1$; When $m=2,3$ and $4$, both $\mu_s(\alpha)$ and $P_s(\alpha)$ are increasing, then the optimal allocation is $\alpha=10$.

\begin{figure}[t]
    \centering
    \includegraphics[width=0.48\textwidth]{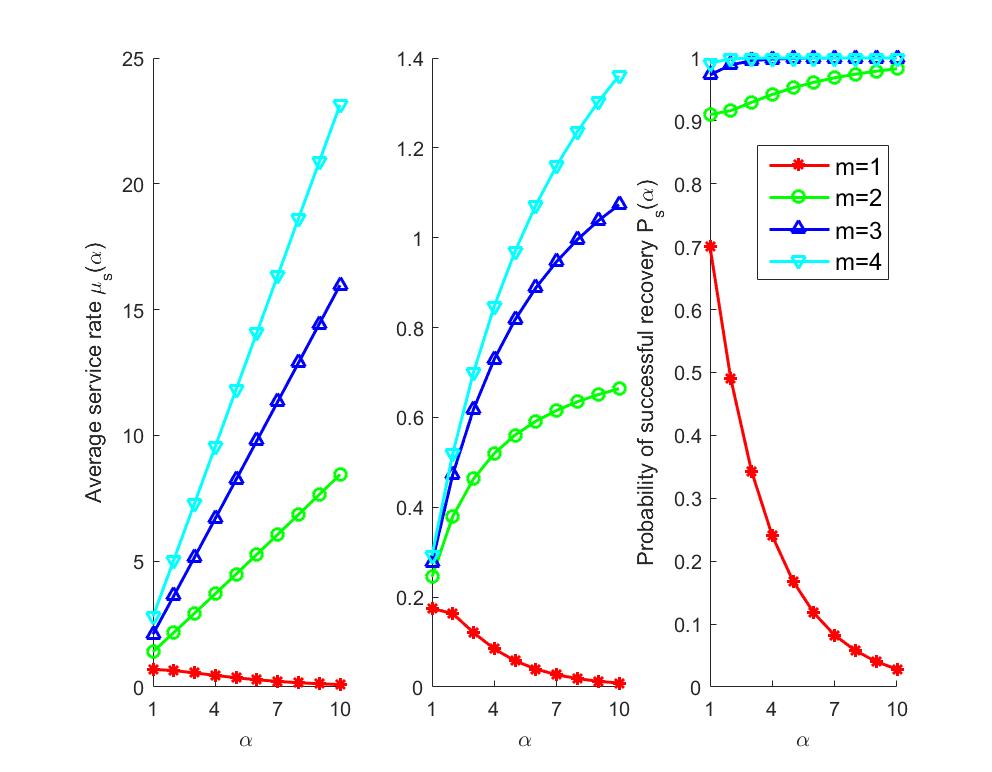}
    \caption{The average service rate and probability of successful recovery of probabilistic access model are changing with the allocation parameter $\alpha$ under different $m$ when $p=0.3$. Left: the service is scaled exponential with $\mu=1$; Middle: the service is shifted exponential with $\mu=1$ and $\Delta=3$; Right: the probability of successful recovery.}
    \label{fig:Prob1}
\end{figure}

In Fig.~\ref{fig:Prob2}, the pattern of $\mu_s(\alpha)$ is changing from increasing to decreasing as $p$ is changing from 0.51 to 0.71. $P_s(\alpha)$ is always decreasing in these three cases. 
 
\begin{figure}[t]
    \centering
    \includegraphics[width=0.48\textwidth]{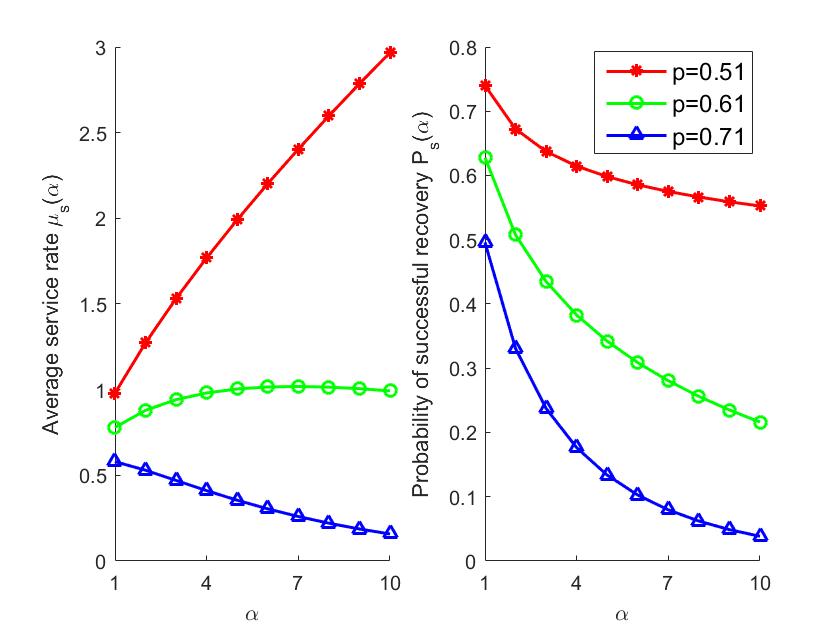}
    \caption{The average service rate and probability of successful recovery of probabilistic access model are changing with the allocation parameter $\alpha$ under different $p$ when $m=2$ and the distribution is scaled exponential with $\mu=1$.}
    \label{fig:Prob2}
\end{figure}

In Fig.~\ref{fig:SProb2}, the slope of $\mu_s(\alpha)$ and $P_s(\alpha)$ switches from increasing to decreasing as $p$ is changing from 0.3 to 0.7.

\begin{figure}[t]
    \centering
    \includegraphics[width=0.48\textwidth]{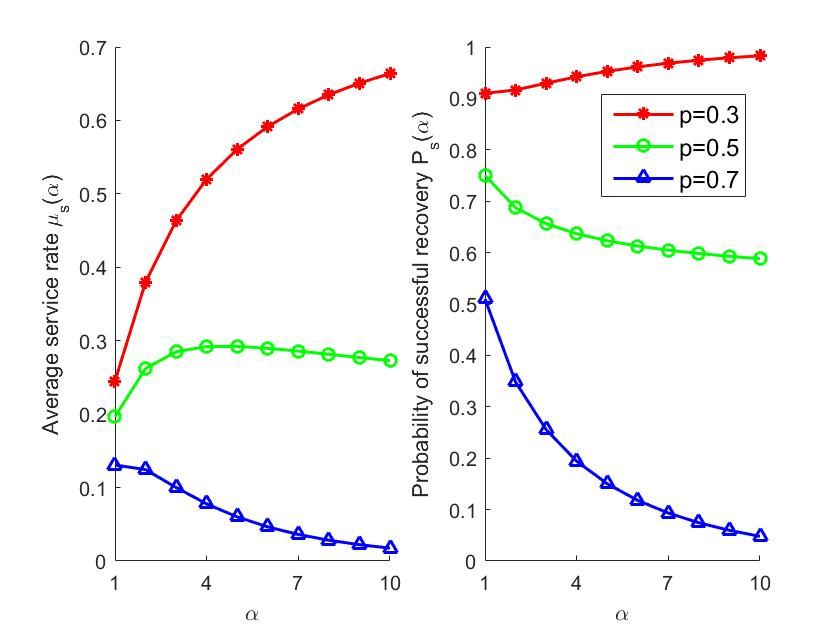}
    \caption{The average service rate and probability of successful recovery of probabilistic access model are changing with the allocation parameter $\alpha$ under different $p$ when $m=2$ and the distribution is scaled exponential with $\mu=1$ and $\Delta=3$.}
    \label{fig:SProb2}
\end{figure}

Here we come to the same conclusion as in the fixed-size access model. The following claim helps understand why the optimal $\alpha$ varies (for proof  see the proof of Claim~\ref{explainfix}):

\begin{claim}
For the probabilistic access model with failure probability for each node $p=0$, if we use scaled exponential distribution, we have $\mu_s(\alpha)$<$\mu_s(\alpha+1)$; If we use shifted exponential distribution, the result varies based on different parameters' values.
\label{explainprob}
\end{claim}

We can see that patterns described by the Conjecture \ref{hyp:fix1} and \ref{hyp:prob1} hold in all the figures. And we can find another pattern according to the coded file size ratio $m$ in Conjecture \ref{hyp:redun}.

\begin{conjecture}
For both fixed-size access and probabilistic access models, when $m$ is increasing, the optimal $\alpha$ for $\mu_s(\alpha)$ is also increasing.
\label{hyp:redun}
\end{conjecture}

\bibliographystyle{IEEEtran}
\bibliography{distributed}

\end{document}